\documentclass[preprint,12pt]{elsarticle}
\usepackage{amssymb}
\usepackage[utf8]{inputenc}

\usepackage{amsthm}
\usepackage{amssymb,latexsym}
\usepackage{amsfonts}
\usepackage{amsmath}
\usepackage{enumerate}
\usepackage{float}
\usepackage{mathtools}
\usepackage{enumerate}
\usepackage{bbm}
\usepackage{bm}
\usepackage{booktabs}

\usepackage[usenames, dvipsnam
es]{xcolor}
\usepackage[colorlinks,linkcolor=blue,anchorcolor=blue,citecolor=blue]{hyperref}

\usepackage{graphicx}

\usepackage{caption}
\usepackage{subcaption}
\usepackage{comment}

\newcommand\R{{\mathbb R}}

\newcommand\al{\alpha}
\newcommand\be{\beta}

\newcommand{\vphi}{\varphi}

\newtheorem{theorem}{Theorem}

\newtheorem{corollary}[theorem]{Corollary}

\newtheorem{proposition}[theorem]{Proposition}

\newtheorem{example}[theorem]{Example}

\usepackage{atbegshi}
\AtBeginDocument{\AtBeginShipoutNext{\AtBeginShipoutDiscard}}
\addtocounter{page}{-1}

\journal{arXiv}

\begin{document}

\begin{frontmatter}
%\title{Overview of interactions between some probability distributions and probability of default}
%\title{Estimating Probabilities of Default for Low
%Default Portfolios Revisited}
%\title{Estimating Probabilities of Default for Low
%Default Portfolios by K. Pluto and D. Tasche Revisited}
\title{Probabilistic Overview of Probabilities of Default for Low
Default Portfolios by K. Pluto and D. Tasche}
\author[inst1]{Andrius Grigutis\footnote{Despite the use of the ''authorial we'', common in academia and meaning the author and the reader, this article is the sole work of its author.}}
\
\affiliation[inst1]{organization={Institute of Mathematics},%Department and Organization
            addressline={Naugarduko 24}, 
            city={Vilnius},
            postcode={LT-03225}, 
            country={Lithuania}}
%\author[inst2]{Author Two}
%\author[inst1,inst2]{Author Three}

%\affiliation[inst2]{organization={Department Two},%Department and Organization
%            addressline={Address Two}, 
%            city={City Two},
%            postcode={22222}, 
%            state={State Two},
%            country={Country Two}}

\begin{abstract}
This article gives a probabilistic overview of the widely used method of default probability estimation proposed by K. Pluto and D. Tasche. There are listed detailed assumptions and derivation of the inequality where the pro\-ba\-bi\-li\-ty of default is involved under the influence of systematic factor. The author anti\-ci\-pates adding more clarity, especially for early career an\-a\-lysts or scholars, regarding the assumption of borrowers' independence, conditional independence and interaction between the pro\-ba\-bi\-lity distributions such as binomial, beta, normal and others. There is also shown the relation between the probability of default and the joint distribution of $\sqrt{\varrho}X-\sqrt{1-\varrho}Y$, where $X$, including but not limiting, is the standard normal, $Y$ admits, including but not limiting, the beta-normal distribution and $X,\,Y$ are independent.
\end{abstract}


\end{frontmatter}

\section{Introduction}\label{sec:introdution}

The {\it probability of default} $p$ is actually the most important metric in credit risk management. Roughly, this probability provides the likelihood for a certain obligor not to follow the taken financial commitments properly within a certain period of time, typically one year. The number of defaulted borrowers divided by the number of total borrowers within a certain portfolio is known as {\it the observed default rate}, while the predicted one $p$ is called {\it the expected default rate}. Any model tasked to predict $p$, should ensure the alignment between the observed and expected default rates. However, in some instances, such as low default portfolios, there is not possible to have any robust observations for the observed default rate. In such instances, the famous work \cite{Pluto2006} suggests applying the Bernoulli trials 
to estimate the experiment's success probability $p$. This work is a survey of two models: (a) the estimation of $p$ when obligors in the portfolio are treated independently of each other and there is no side influence for such a portfolio, (b) the estimation of $p$ when obligors in the portfolio are treated conditionally independent of each other when each obligor is influenced by some systematic factor. We aim to reflect on the detailed steps and assumptions used in deriving the two mentioned models.   

Let us recall several well-known probability distributions.
\begin{itemize}
\item We say that the random variable $X$ is binomial distributed with pa\-ra\-me\-ters $n\in\mathbb{N}$ and $p\in(0,\,1)$ (denoted $X\sim\mathcal{B}in(n,\,p)$) if the probability mass function is
$$
\mathbb{P}(X=k)=\binom{n}{k}p^k(1-p)^{n-k},\,k=0,\,1,\,\ldots,\,n,
$$
where 
$$
\binom{n}{k}=\frac{n!}{k!(n-k)!}.
$$
We denote the cumulative distribution function of the binomial random variable by  
$$
\mathcal{B}in_{n,\,p}(k):=\mathbb{P}(X\leqslant k)=\sum_{i=0}^{k}\binom{n}{i}p^i(1-p)^{n-i},\,
k=0,\,1,\,\ldots,\,n.
$$
\item We say that the random variable $X$ is beta distributed with parameters $\alpha>0$ and $\beta>0$ (denoted $X\sim\mathcal{B}(\alpha,\,\beta)$) if its probability density function is 
$$
b_{\al,\,\be}(x):=\frac{x^\alpha(1-x)^\beta}{B(\alpha,\,\beta)},\,x\in(0,\,1),
$$
where
$$
B(\al,\,\be)=\frac{\Gamma(\al)\Gamma(\be)}{\Gamma(\al+\be)}
$$
and the gamma function for the complex number $s\in\mathbb{C}$ is
$$
\Gamma(s)=\int_{0}^{\infty}t^{s-1}e^{-t}dt,\,\Re s>0.
$$

We denote the cumulative distribution function of the beta random variable by  
$$
\mathcal{B}_{\al,\,\be}(x)=\int_{0}^{x}b_{\al,\,\be}(y)dy,\,x\in(0,\,1)
$$
and its inverse by $\mathcal{B}^{-1}_{\al,\,\be}(x),\,x\in(0,\,1)$.

\item We say that the random variable $X$ is normally distributed with pa\-ra\-me\-ters $\mu\in\R$ and $\sigma^2>0$ (denoted by $X\sim\mathcal{N}(\mu,\,\sigma^2)$) if its probability density function is 
$$
\vphi_{\mu,\,\sigma^2}(x)=\frac{1}{\sigma\sqrt{2\pi}}e^{-(x-\mu)^2/2\sigma^2},\,x\in\mathbb{R}.
$$
We denote the cumulative distribution function of the normal distribution by

$$
\Phi_{\mu,\,\sigma^2}(x)=\int_{-\infty}^{x}\vphi_{\mu,\,\sigma^2}(y)dy,\,x\in\mathbb{R}
$$
and its inverse by $\Phi^{-1}_{\mu,\,\sigma^2}(x)$. We recall the symmetry $\Phi_{\mu,\,\sigma^2}(x)=1-\Phi_{\mu,\,\sigma^2}(-x+2\mu)$ and write $\vphi(x)$, $\Phi(x)$ and $\Phi^{-1}(x)$ respectively if $X\sim\mathcal{N}(0,\,1)$. 
\end{itemize}

The other distributions met in this paper, are introduced in the proper places where they are used.

\section{Binomial and mixture binomial distributions for the probability of default estimation}\label{section_2}

In this section, in the subsections \ref{sub1} and \ref{sub2} respectively, we review the derivation of two methods used to estimate the probability of default $p$. As mentioned in Introduction \ref{sec:introdution}, the first method is just the Bernoulli trials assuming the obligors' independence, while the second method provides the estimation of $p$ under the assumption of obligors' conditional independence of each other under the influence of a certain systematic factor.

\subsection{Binomial distribution}\label{sub1}

Let $X_1,\,X_2,\,\ldots,\,X_n$ be independent copies of Bernoulli random variable $X$ which distribution is $\mathbb{P}(X=1)=p=1-\mathbb{P}(X=0)$. In risk management, the random variables $X_1,\,X_2,\,\ldots,\,X_n$ are treated as independent obligors and the attained value $X_i=1$, $i=1,\,2,\,\ldots,\,n$ means that the $i$'th obligor defaults within some observation period (typically one year), while $X_i=0$, $i=1,\,2,\,\ldots,\,n$ means that the $i$'th obligor does not default within the same observation period. Then, the sum $Y:=X_1+X_2+\ldots+X_n$ may attain any value form the set $\{0,\,1,\,\ldots,\,n\}$ with probability
\begin{align}\label{binom_m}
\mathbb{P}(Y=k)=\binom{n}{k}p^k(1-p)^{n-k},\,
k=0,\,1,\,\ldots,\,n.
\end{align}

The probability mass function of the binomial distribution \eqref{binom_m} means the probability to default $k$ out of total $n$ obligors in the portfolio, while the distribution function \eqref{binom_d} 
\begin{align}\label{binom_d}
\mathcal{B}in_{n,\,p}(k)=\mathbb{P}(Y\leqslant k)=\sum_{i=0}^{k}\binom{n}{i}p^i(1-p)^{n-i},\,k=0,\,1,\,\ldots,\,n
\end{align}
is the probability to default no more than $k$ obligors out of total $n$. Obviously, 
$$
\mathcal{B}in_{n,\,p}(n)=\mathbb{P}(Y\leqslant n)=\sum_{i=0}^{n}\binom{n}{i}p^i(1-p)^{n-i}=(p+1-p)^n=1
$$
for any $p\in(0,1)$. 

In many examples, e.g., tossing a coin or a die, the experiment's success probability $p$ is known beforehand. However, in real-life problems, such as default probability estimation, the probability $p$ is desired to know. In order to get $p$ out of \eqref{binom_m} or \eqref{binom_d} we need an expert judgment first. Let us suppose that the probability of the amount of defaulted obligors does not exceed $k\in\{0,\,1,\,\ldots,\,n-1\}$ out of $n$ is at least $1-\gamma$. Then, 
\begin{align}\label{bin_ineq}\nonumber
&\mathcal{B}in_{n,\,p}(k)=\mathbb{P}(Y\leqslant k)\\
&=\sum_{i=0}^{k}\binom{n}{i}p^i(1-p)^{n-i}\geqslant (1-\gamma),\,
k\in\{0,\,1,\,\ldots,\,n-1\}
\end{align}
and, in view of Proposition \ref{beta_prop}, the upper bound of default probability $p$ is
\begin{align}\label{p_est_no_syst}
p\leqslant 1-\mathcal{B}^{-1}_{n-k,\,k+1}(1-\gamma),
\end{align}
where $\mathcal{B}^{-1}_{n-k,\,k+1}(\cdot)$ is an inverse of beta distribution function. In particular, if $k=0$, i.e., we are certain with probability $1-\gamma$ that there be no defaulted obligors at all, then
$$
p\leqslant 1 -(1-\gamma)^{1/n}.
$$

The probability $1-\gamma$ can be introduced as the probability of type I error, also known as the false positive instance classification, which in our context means that the actual probability of default does not belong to the predicted interval
$0\leqslant p\leqslant 1-\mathcal{B}^{-1}_{n-k,\,k+1}(1-\gamma)$; see \cite{beta_binom}. Moreover, the confidence interval of the binomial distribution is known as the Clopper–Pearson interval; see \cite{CLopper_Pearson}, \cite{CLopper_Pearson2}.
 
\subsection{Mixture of Binomial and Normal distributions}\label{sub2}

Let $r$ be an annual return rate and $(1+r/n)^n,\,n\in\mathbb{N}$ the increment of the invested amount when the return rate is compounded $n$ times per year. It is well known that $(1+r/n)^n\to e^r$ when $n\to\infty$. Thus, 
\begin{align*}
V_F=V_Ie^r
\end{align*}
assuming the continuously compounded return, where $V_{F}>0$ denotes the final value and $V_{I}>0$ the initial one. Based on the previous thoughts, we define
\begin{align}\label{log_return}
r_{log}:=\ln \frac{V_{F}}{V_{I}}=\ln V_F-\ln V_I.
\end{align}
The return derived in \eqref{log_return} is called the {\it logarithmic return} or just {\it log-return}. We now assume the logarithmic return to be a random variable. More precisely, we assume 
\begin{align}\label{log_r_rand}
r_{log}=\beta S+\xi,
\end{align}
where $\beta\in\mathbb{R}$, $S\sim\mathcal{N}(\mu_1,\,\sigma^2_1)$, $\xi\sim\mathcal{N}(\mu_2,\,\sigma^2_2)$, the random variables $S$ and $\xi$ are independent and both non-degenerate. Also, $S$ is known as {\it systematic} risk factor, while $\xi$ as {\it idiosyncratic}; see \cite{syst_idio}. The origin of return's definition \eqref{log_r_rand} has similarities with the {\it capital asset pricing model} which states that every expected return $\mathbb{E}r_i$ under certain assumptions satisfies
$$
\mathbb{E}(r_i-r_f)=\beta_i\mathbb{E}(r_M-r_f),
$$
where $r_f$ is the risk-free return rate, $r_M$ the return of systemic portfolio $M$ and $\beta_i=cov(r_i,\,r_M)/\sigma^2_M$, see, for example, \cite{CAPM} and observe that \eqref{log_r_rand} implies $\mathbb{E}(r_{log}-\xi)=\beta\mathbb{E}S$.

Let us now standardize the log-return \eqref{log_r_rand}. It is easy to check that
$$
r_{log}=\beta S+\xi \quad \Leftrightarrow \quad
\frac{r_{log}-\mathbb{E}r_{log}}{\sigma_{r_{log}}}=
\beta\frac{\sigma_S}{\sigma_{r_{log}}}\frac{S-\mathbb{E}S}{\sigma_S}
+\frac{\sigma_{\xi}}{\sigma_{r_{log}}}\frac{\xi-\mathbb{E}\xi}{\sigma_{\xi}}.
$$
Thus, it is equivalent to define $r_{log}$ as
\begin{align}
\tilde{r}_{log}=\sqrt{\varrho}\tilde{S}+\sqrt{1-\varrho}\tilde{\xi},\,\varrho\in[0,1],
\end{align}
where 
\begin{align}\label{rho}
\varrho=\left(\beta\frac{\sigma_S}{\sigma_{r_{log}}}\right)^2
=\frac{\beta^2 \sigma_S^2}{\beta^2 \sigma^2_S+\sigma^2_{\xi}}
, 
\end{align}
and $\tilde{S}$, $\tilde{\xi}$ are independent standard normal random variables. Indeed, $r_{log}\sim\mathcal{N}(\beta\mu_1+\mu_2,\,\beta^2\sigma_1^2+\sigma_2)$ is quivalent to 
$\tilde{r}_{log}\sim\mathcal{N}(0,\,1)$.

We note that
$$
1=\sigma^2_{\tilde{r}_{log}}=\left(\beta\frac{\sigma_S}{\sigma_{r_{log}}}\right)^2+\left(\frac{\sigma_{\xi}}{\sigma_{r_{log}}}\right)^2
$$
and the coefficient $\varrho$ in \eqref{rho} is called the {\it asset correlation} (see \cite{asset_cor}); it expresses the correlation between $\tilde{r}_{log}$ and $\tilde{S}$:
$$
corr(\tilde{r}_{log},\,\tilde{S})=cov(\sqrt{\varrho}\tilde{S}+\sqrt{1-\varrho}\tilde{\xi},\,\tilde{S})=\sqrt{\varrho}.
$$

We now define the default event $D$ by

\begin{align}\label{D}
D=
\begin{cases}
&1,\textit{ if } \sqrt{\varrho}\tilde{S}+\sqrt{1-\varrho}\tilde{\xi}<x_p, \\
&0, \textit{ otherwise}.
\end{cases}
\end{align}
Of course, $D$ is Bernoulli random variable and $x_p=\Phi^{-1}(p)$ since the random variable $\sqrt{\varrho}\tilde{S}+\sqrt{1-\varrho}\tilde{\xi}$ is standard normal.  We now are interested in that particular $p$ which causes $D=1$. Conditioning on $\tilde{S}$, i.e., assuming that the systematic factor attains some particular value $x\in\mathbb{R}$, for $\varrho\neq1$, we have
\begin{align}\label{PD}\nonumber
&\mathbb{P}(D=1|\tilde{S}=x)\\
&=\mathbb{P}\left(\tilde{\xi}<\frac{\Phi^{-1}(p)-\sqrt{\varrho}\tilde{S}}{\sqrt{1-\varrho}}\Bigg| \tilde{S}=x\right)
=\Phi\left(\frac{\Phi^{-1}(p)-\sqrt{\varrho}x}{\sqrt{1-\varrho}}\right)
\end{align}
and
\begin{align}\label{PD0}
\mathbb{P}(D=0|\tilde{S}=x)=1-\Phi\left(\frac{\Phi^{-1}(p)-\sqrt{\varrho}x}{\sqrt{1-\varrho}}\right).
\end{align}
The random variable 
\begin{align}\label{Vasicek}
\Phi\left(\frac{\Phi^{-1}(p)-\sqrt{\varrho}\tilde{S}}{\sqrt{1-\varrho}}\right),\,p\in(0,1),\,\varrho\in[0,1),
\end{align}
where $\tilde{S}\sim\mathcal{N}(0,\,1)$, is known as {\it Vasicek distribution}, see \cite{Vasicek}.

Let $D_1,\,D_2,\,\ldots,D_n$ be the conditionally independent copies of the random variable $D$ when the systematic factor $\tilde{S}=x$. Then, $\mathcal{D}:=D_1+D_2+\ldots+D_n$ is binomial random variable and the conditional probability that $\mathcal{D}=i$ if $\tilde{S}=x$ is
\begin{align}\label{bin_dist_s}\nonumber
&\mathbb{P}(\mathcal{D}=i|\tilde{S}=x)\\
&=\binom{n}{i}
\Phi^i\left(\frac{\Phi^{-1}(p)-\sqrt{\varrho}x}{\sqrt{1-\varrho}}\right)
\left(1-\Phi\left(\frac{\Phi^{-1}(p)-\sqrt{\varrho}x}{\sqrt{1-\varrho}}\right)\right)^{n-i},
\end{align}
where $i=0,\,1,\,\ldots,\,n$. Thus, being certain with probability at least $1-\gamma$, that there default up to $k\in\{0,\,1,\,\ldots,\,n-1\}$ obligors out of total $n$, by the law of total probability we get
\begin{align}\label{main_ineq}\nonumber
&\mathbb{P}(\mathcal{D}\leqslant k)=\mathbb{E}\left(\mathbb{P}(\mathcal{D}\leqslant k|\tilde{S}=x)\right)=\\ \nonumber
&\int_{-\infty}^{+\infty}\varphi(x)\sum_{i=0}^{k}\binom{n}{i}\left(\Phi\left(\frac{\Phi^{-1}(p)-\sqrt{\varrho}x}{\sqrt{1-\varrho}}\right)\right)^i\left(1-\Phi\left(\frac{\Phi^{-1}(p)-\sqrt{\varrho}x}{\sqrt{1-\varrho}}\right)\right)^{n-i}\hspace{-0.5cm}dx\\
&\geqslant 1-\gamma.
\end{align}
Notice that if $k=n$, then the inequality \eqref{main_ineq} is satisfied with any $p\in(0,\,1)$ when $\gamma\in[0,\,1]$. Also, $\varrho=0$ in \eqref{main_ineq} implies the inequality \eqref{bin_ineq}. Equally, the integral in \eqref{main_ineq} is nothing but the mixture of the binomial and Vasicek distributions: it is the cumulative binomial distribution function $\mathcal{B}in_{n,\,p}(k),\,k=0,\,1,\,\ldots,\,n$ when the parameter $p$ is Vasicek distributed \eqref{Vasicek}. See \cite{mixture} for the mixture distribution models.

According to Proposition \ref{main_prop}, the upper bound of $p$ in \eqref{main_ineq} is
\begin{align}\label{p_est_syst}
p\leqslant 1-\Phi\left(\sqrt{1-\varrho}F_{n-k,\,k+1,\,\varrho}^{-1}(1-\gamma)\right), 
\end{align}
where $F^{-1}_{n-k,\,k+1,\,\varrho}(\cdot)$ is the inverse of the cumulative distribution function
\begin{align}\label{cum_dist}
F_{n-k,\,k+1,\,\varrho}(y)=\int_{0}^{1}\mathcal{B}_{n-k,\,k+1}\left(\Phi\left(\sqrt{\frac{\varrho}{1-\varrho}}\Phi^{-1}(x)+y\right)\right)dx,\,y\in\mathbb{R}.
\end{align}

It is not easy to get a more convenient expression of the cumulative distribution function $F_{n-k,\,k+1,\,\varrho}(y)$ in \eqref{cum_dist}. Thus, we should search for the quantiles of the underlying distribution, described by $F_{n-k,\,k+1,\,\varrho}(y)$, numerically; see Section \ref{sec:examples}. Of course, the function $F_{n-k,\,k+1,\,\varrho}(y)$ is defined in view of Proposition \ref{main_prop} by replacing 
$$
y=-\frac{\Phi^{-1}(p)}{\sqrt{1-\varrho}}
$$
in \eqref{eq:second} and there is equivalent to search for such $p\in(0,1)$ that

$$
F_{n-k,\,k+1,\,\varrho}\left(-\frac{\Phi^{-1}(p)}{\sqrt{1-\varrho}}\right)\geqslant 1-\gamma
$$
or
\begin{align}\label{dist_f}
\tilde{F}_{n-k,\,k+1,\,\varrho}(p):=1-F_{n-k,\,k+1,\,\varrho}\left(-\frac{\Phi^{-1}(p)}{\sqrt{1-\varrho}}\right)\leqslant\gamma,
\end{align}
where $\tilde{F}_{n-k,\,k+1,\,\varrho}(p),\, p\in(0,1)$ is the continuous cumulative distribution function with respect to $p$.

Let us mention that the probability distribution, described by its cumulative distribution function
$$
\mathcal{BN}_{\alpha,\,\beta,\,\mu,\,\sigma^2}(x):=\mathcal{B}_{\alpha,\,\beta}\left(\Phi_{\mu,\,\sigma^2}(x)\right),\,x\in \mathbb{R},
$$
is known as beta-normal. We write $X\sim\mathcal{BN}(\alpha,\,\beta,\,\mu,\,\sigma^2)$ if $X$ is the beta-normal random variable and $bn_{\alpha,\,\beta,\,\mu,\,\sigma^2}(x)$, $x\in\mathbb{R}$ denotes its density. See \cite{BN1}, \cite{BN2}, \cite{BN3} and \cite{Beta_moments} for the beta-normal distribution. Thus, $F_{\al,\,\be,\,\varrho}(y)$ in \eqref{cum_dist} can be easily described in terms of the beta-normal distribution. See also \cite{Bluhm} as the good initial source on credit risk management and some other insights deriving inequality \eqref{main_ineq}. Equally, in view of \eqref{cum_dist}, we depict the probability density function
$$
\frac{d\,F_{\alpha,\,\beta,\,\varrho}(y)}{dy},\,y\in\mathbb{R},\,\alpha>0,\,\beta>0,\,0\leqslant\varrho<1
$$
for some chosen parameters in Figure \ref{fig} and the cumulative distribution function $F_{\alpha,\,\beta,\,\varrho}(y)$ itself correspondingly in Figure \ref{fig_1} below.
\begin{figure}[H]
\includegraphics[scale=0.75]{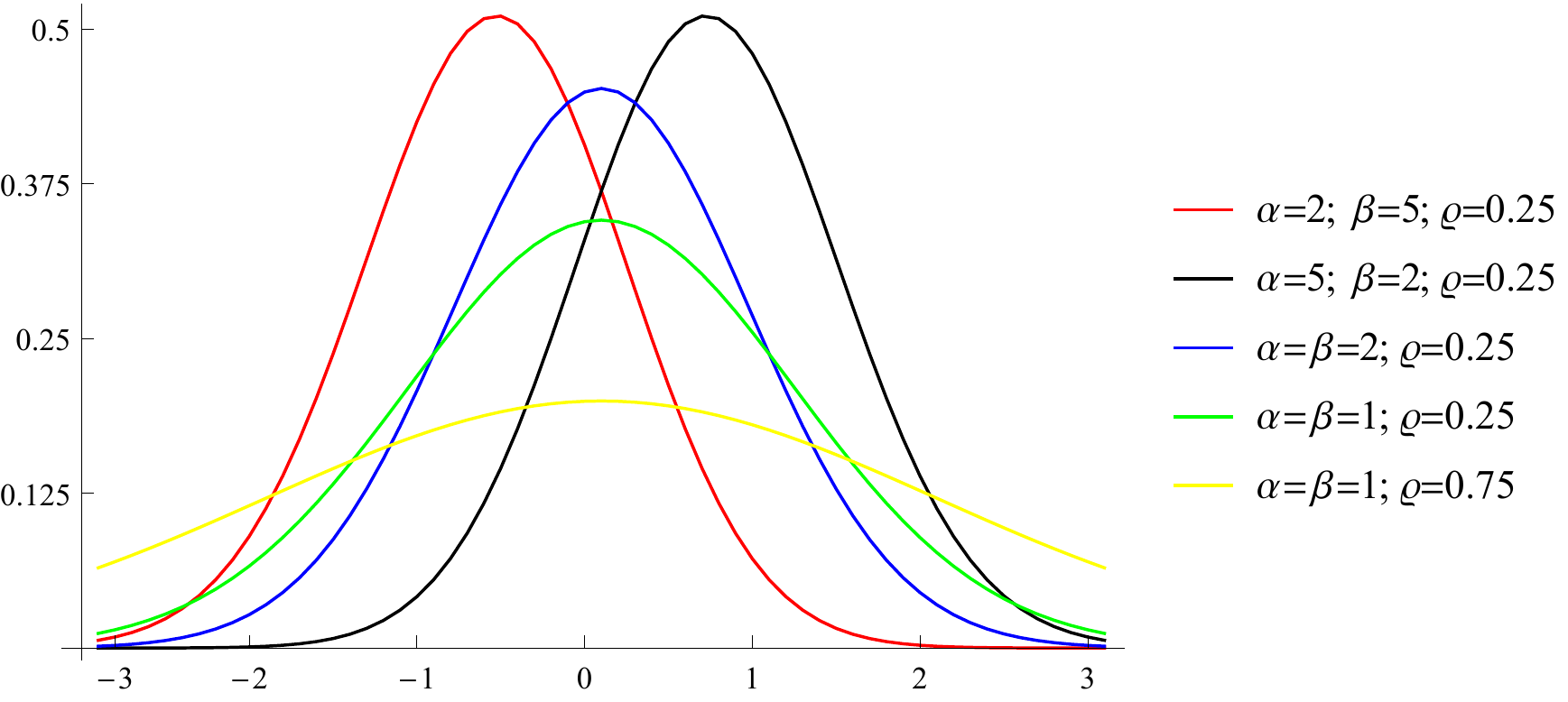}
\caption{The probability density function whose cumulative distribution function is $F_{\alpha,\,\beta,\,\varrho}(y)$, $y\in\mathbb{R}$.}\label{fig}
\end{figure}

\begin{figure}[H]
\includegraphics[scale=0.75]{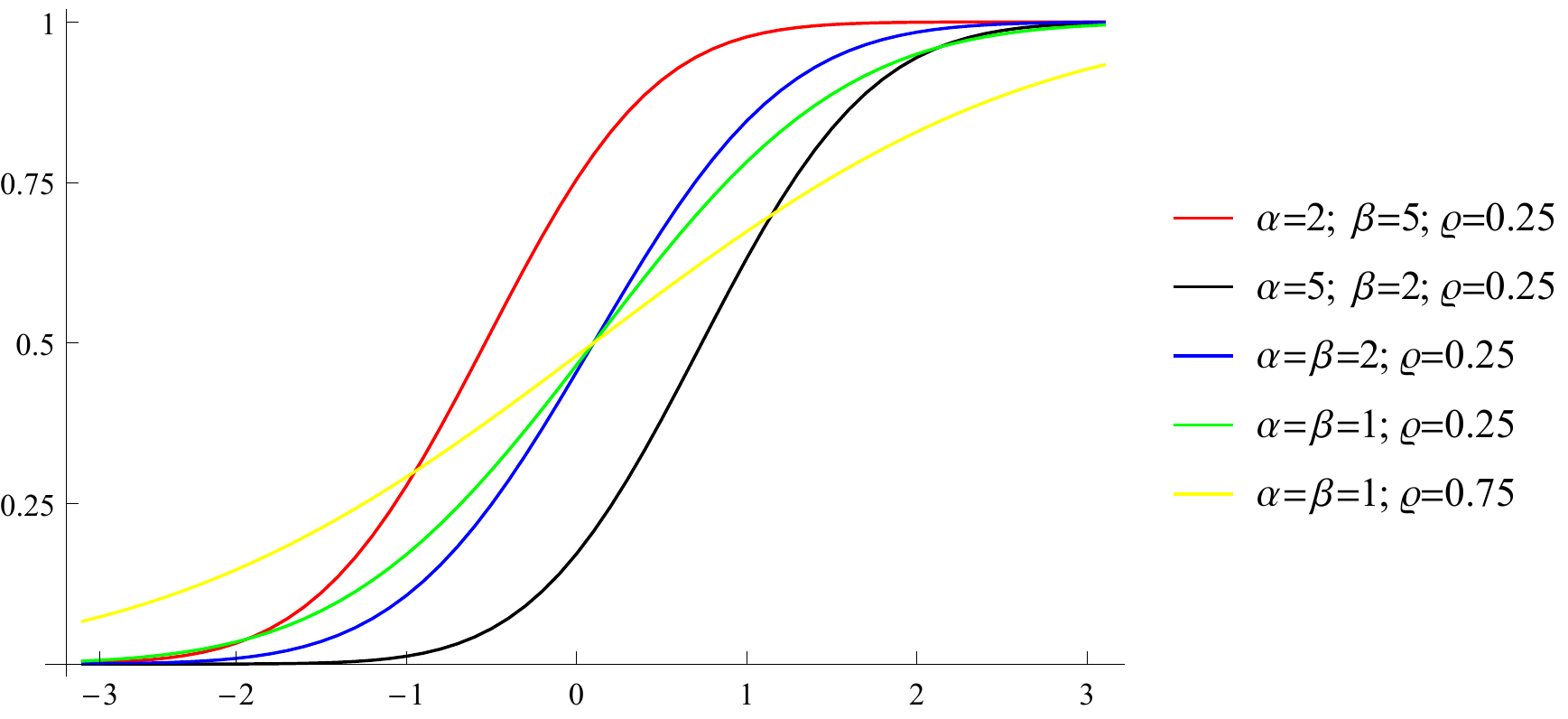}
\caption{The cumulative distribution function $F_{\alpha,\,\beta,\,\varrho}(y)$, $y\in\mathbb{R}$.}\label{fig_1}
\end{figure}

The derivative of $\tilde{F}_{n-k,\,k+1,\,\varrho}(p)$ in \eqref{dist_f} and the cumulative distribution function $\tilde{F}_{n-k,\,k+1,\,\varrho}(p)$ itself for some chosen parameters are depicted in Figure \ref{fig_2} and Figure \ref{fig_3} below respectively.

\begin{figure}[H]
\includegraphics[scale=0.7]{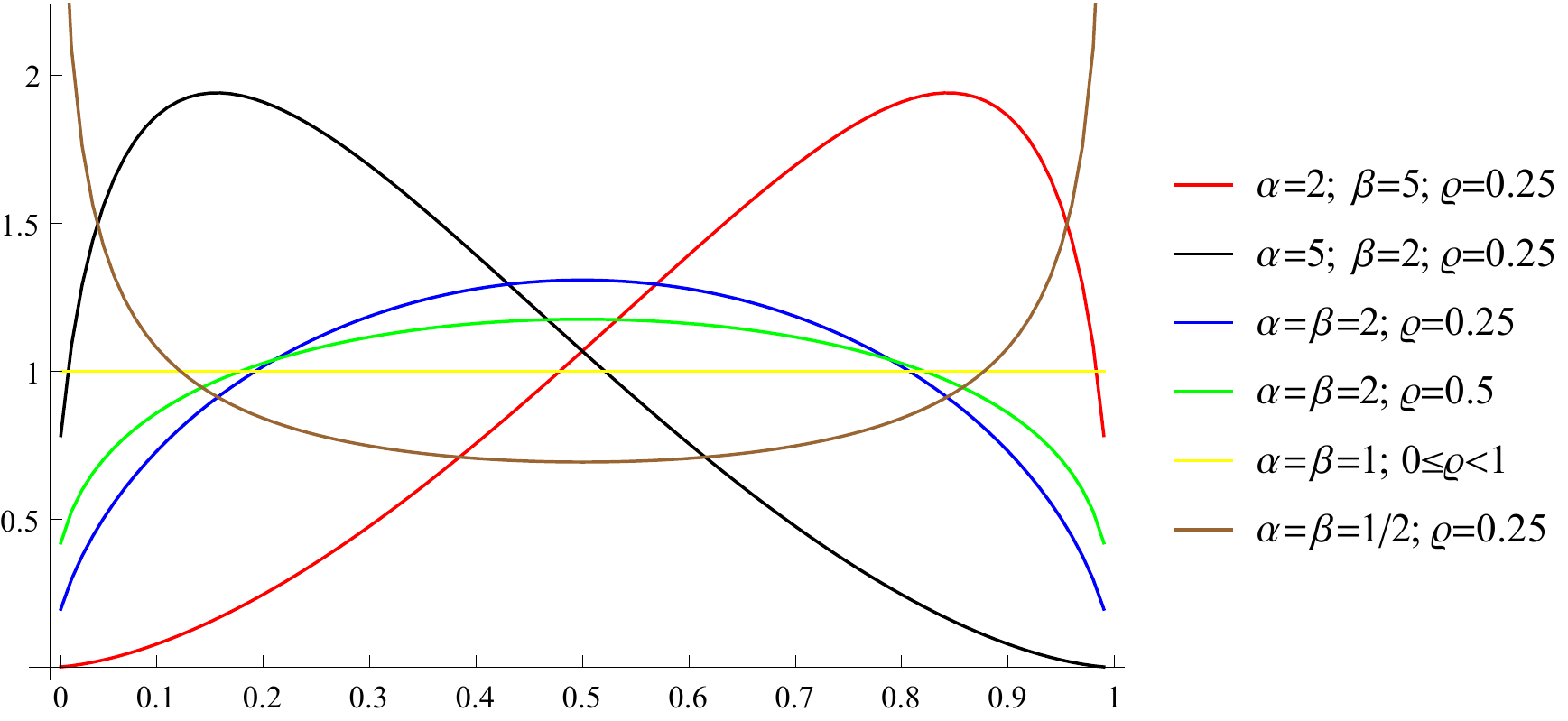}
\caption{The probability density function whose cumulative distribution function is $\tilde{F}_{n-k,\,k+1,\,\varrho}(p)$, $p\in(0,\,1)$.}\label{fig_2}
\end{figure}

\begin{figure}[H]
\includegraphics[scale=0.75]{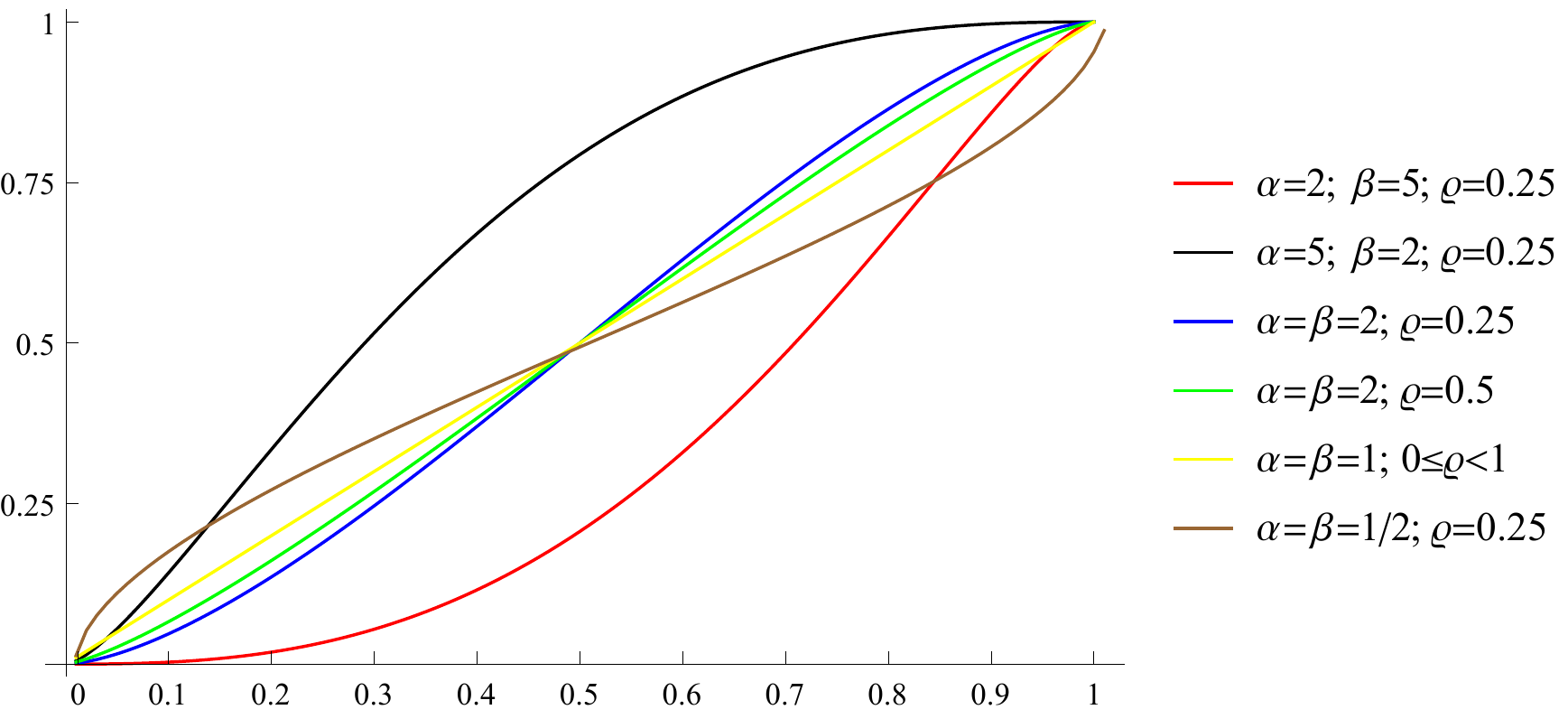}
\caption{The cumulative distribution function $\tilde{F}_{n-k,\,k+1,\,\varrho}(p)$, $p\in(0,\,1)$.}\label{fig_3}
\end{figure}

The depicted functions in Figures \ref{fig}-\ref{fig_3} relate the search of the upper bound of $p$ with the quantile function of $1-\gamma$ under the underlying distribution. According to Propositions \ref{main_prop} and \ref{prop:dist_eq}, in Figures \ref{fig_4}, \ref{fig_5} and \ref{fig_6} below we illustrate the search of the upper bound of $p$ relation with the partial volume of the unit given by $1-\gamma$ (the blue colored volume in Figures \ref{fig_4}, \ref{fig_5} and \ref{fig_6}) under the joint density surface.  

\begin{figure}[H]
\begin{center}
\includegraphics[scale=0.7]{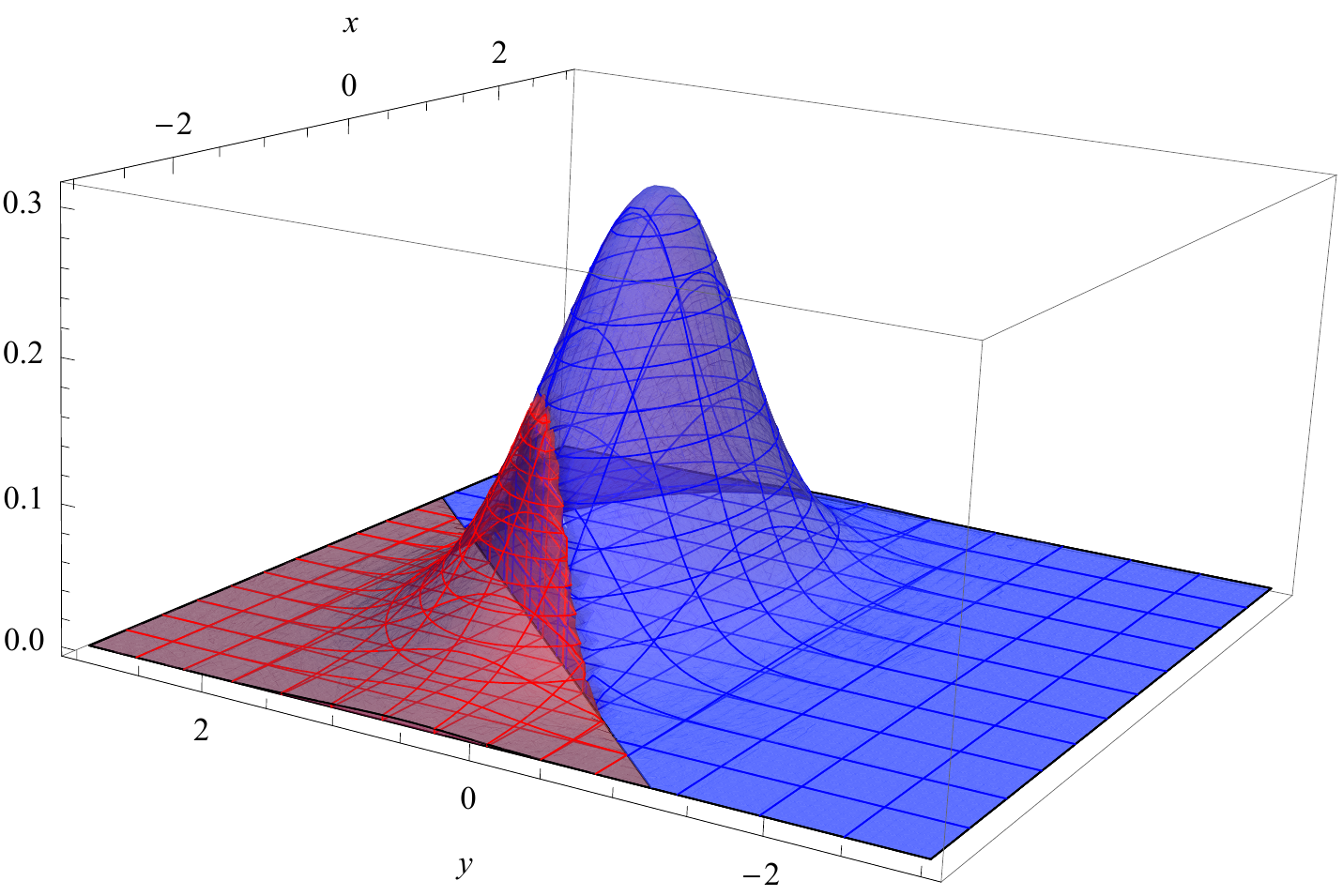}
\end{center}
\caption{The joint density $\varphi(x)bn_{\alpha,\,\beta,\,0,\,1}(y)$, $(x,\,y)\in\mathbb{R}^2$ and the line $\sqrt{\varrho}x-\sqrt{1-\varrho}y=\Phi^{-1}(p)$, when $\alpha=5$, $\beta=2$, $\varrho=1/2$ and $p=1/10$. The red colored volume corresponds to $\sqrt{\varrho}x-\sqrt{1-\varrho}y<\Phi^{-1}(p)$, while the blue one is $1-\gamma=0.869$. In other words, the inequality \eqref{main_ineq} with $1-\gamma=0.869$, $k=1$, $n=6$ and $\varrho=1/2$ is satisfied when $p\in(0,\,1/10]$.}\label{fig_4}
\end{figure}
\vspace{-0.5cm}
\begin{figure}[H]
\begin{center}
\includegraphics[scale=0.7]{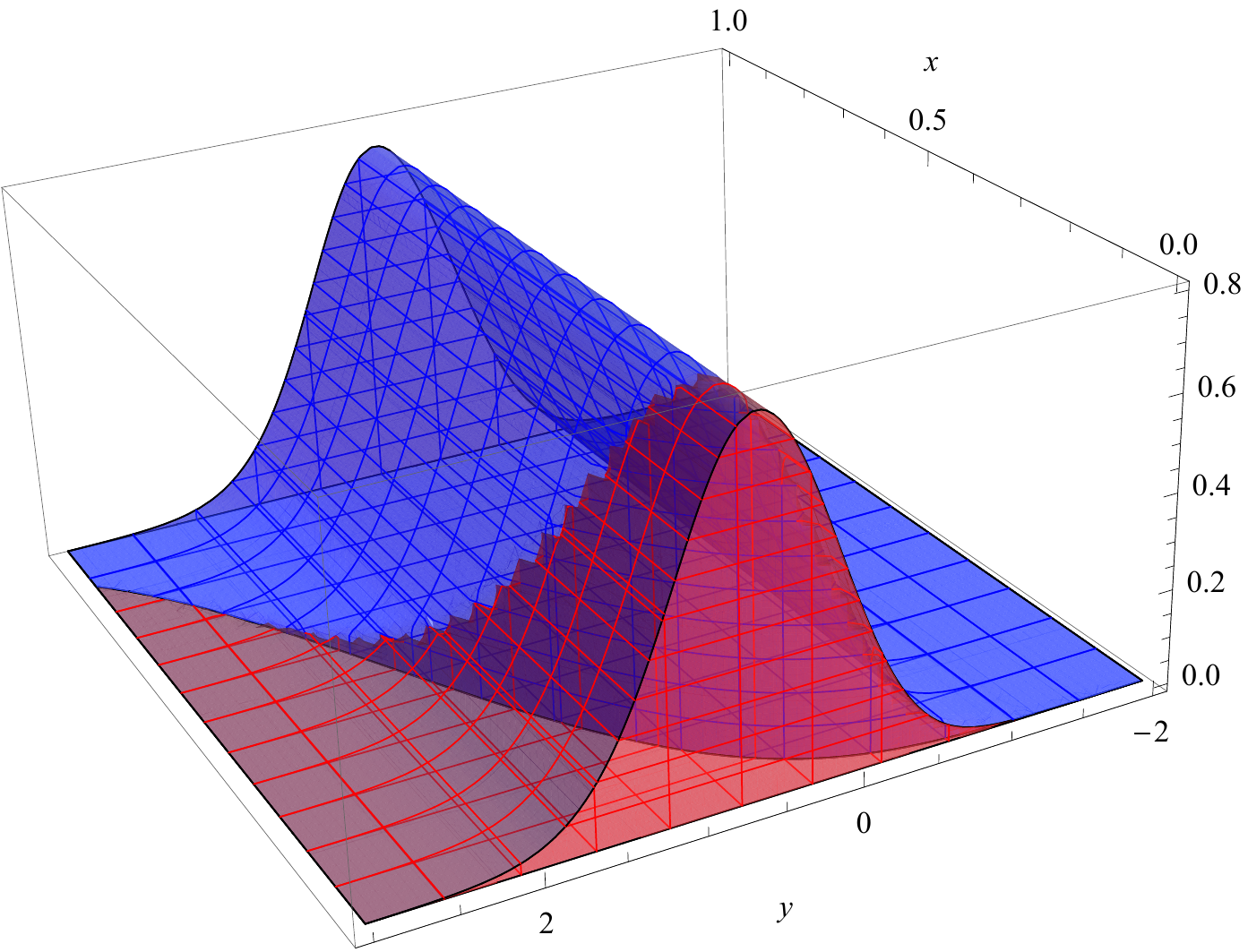}
\end{center}
\caption{The joint density $1\cdot bn_{\alpha,\,\beta,\,0,\,1}(y)$, $0<x<1$, $y\in\mathbb{R}$ and the curve $\sqrt{\varrho}\Phi^{-1}(x)-\sqrt{1-\varrho}y=\Phi^{-1}(p)$, when $\alpha=5$, $\beta=2$, $\varrho=1/2$ and $p=1/10$. The red colored volume corresponds to $\sqrt{\varrho}\Phi^{-1}(x)-\sqrt{1-\varrho}y<\Phi^{-1}(p)$, while the blue one is $1-\gamma=0.869$.}\label{fig_5}
\end{figure}

\begin{figure}[H]
\begin{center}
\includegraphics[scale=0.7]{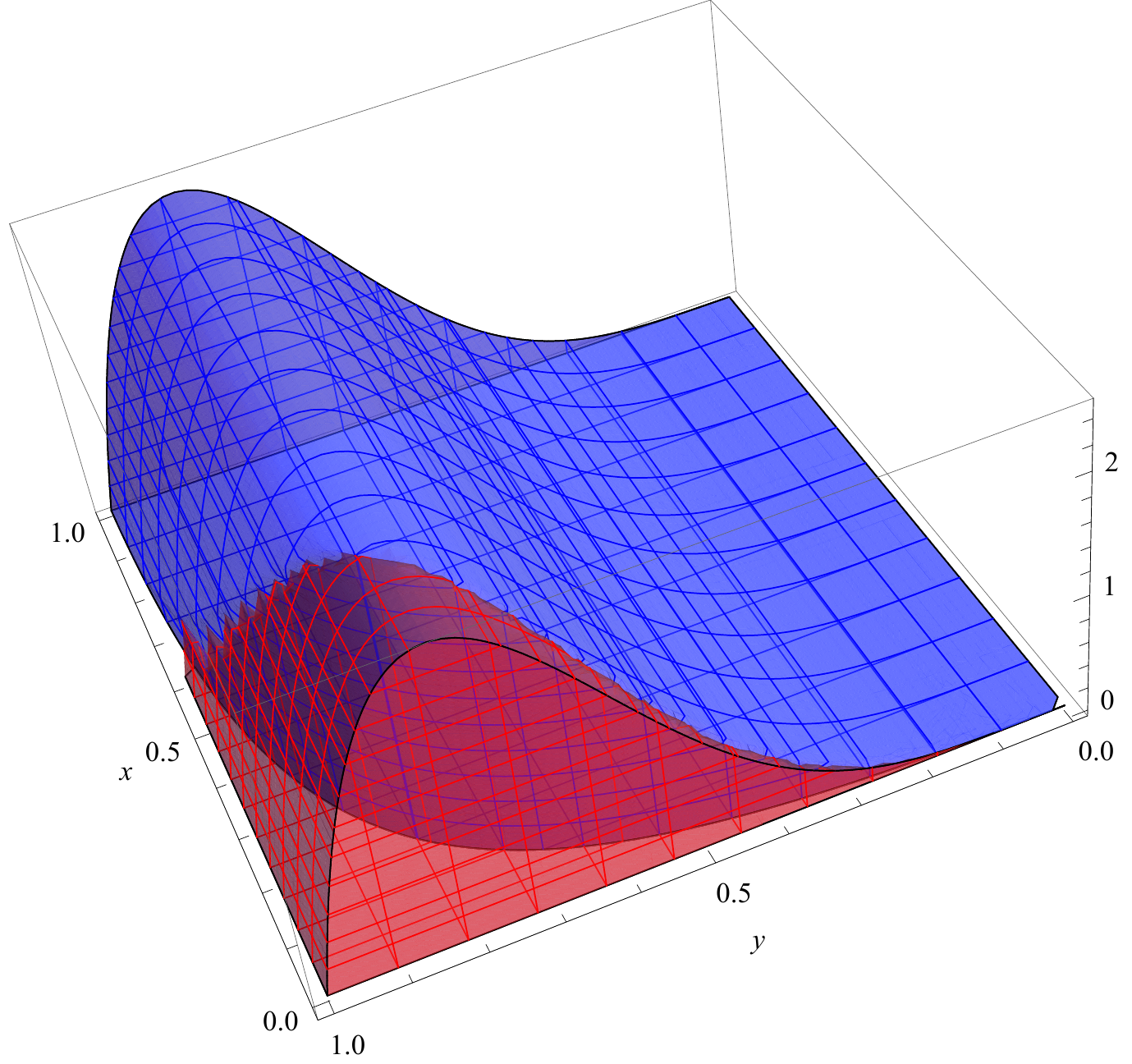}
\end{center}
\caption{The joint density $1\cdot b_{\alpha,\,\beta}(y)$, $0<x,\,y<1$, and the curve $\sqrt{\varrho}\Phi^{-1}(x)-\sqrt{1-\varrho}\Phi^{-1}(y)=\Phi^{-1}(p)$, when $\alpha=5$, $\beta=2$, $\varrho=1/2$ and $p=1/10$. The red co\-lo\-red volume corresponds to $\sqrt{\varrho}\Phi^{-1}(x)-\sqrt{1-\varrho}\Phi^{-1}(y)<\Phi^{-1}(p)$, while the blue one is $1-\gamma=0.869$.}\label{fig_6}
\end{figure}

To estimate the probability of default $p$ by \eqref{p_est_no_syst} or \eqref{p_est_syst} among the portfolio sub-classes $A_1,\,A_2,\,\ldots,\,A_l$, where $A_1$ represents the lowest risk borrowers and $A_l$ the highest respectively, there was proposed a {\it method of conservatism}; see \cite{Pluto2006}. The method of conservatism states the following. Let $n_1,\,n_2,\,\ldots,\,n_l$, $k_1,\,k_2,\,\ldots,\,k_l$ and $p_1,\,p_2,\,\ldots,\,p_l$ be the number of obligors, the number of expected defaults and default probabilities over the portfolio sub-classes $A_1,\,A_2,\,\ldots,\,A_l$ respectively. Then $n_1+n_2\ldots+n_l=n$, $k_1+k_2+\ldots +k_l=k$ and the probability of defaults $p_1$ should be estimated using the parameters $(n,\,k)$ in \eqref{p_est_no_syst} or \eqref{p_est_syst}, $p_2$ should be estimated using $(n-n_1,\,k-k_1)$, $p_3$ with $(n-n_1-n_2,\,k-k_1-k_2)$ and so on up to $p_l$ which should be estimated using $(n_l,\,k_l)$.

Discussions and dissatisfaction among the practitioners that the estimates \eqref{p_est_no_syst} or \eqref{p_est_syst} of the probability of default are too conservative, force some adjustments to estimate $p$ conditionally (biased), see \cite{Dirk} and related papers.

\section{Statements}\label{sec:statements}

In this section, we recall the connection between the binomial and beta distributions, provide several equivalent forms of inequality \eqref{main_ineq} and its connection to the normal multivariate distribution when there are no expected defaults, i.e. $k=0$.

\begin{proposition}\label{beta_prop}
Let $n\in\mathbb{N}$, $k\in\{0,\,1,\,\ldots,\,n-1\}$ be fixed and $p\in(0,\,1)$. Then the cumulative distribution function of binomial and beta random variables are related as
$$
1-\mathcal{B}_{k+1,n-k}(p)=\mathcal{B}_{n-k,k+1}(1-p)=\mathcal{B}in_{n,p}(k),\,p\in(0,\,1).
$$
\end{proposition}

{\sc Note 1}: Let us emphasize that the function $\mathcal{B}in_{n,p}(k)$ in Proposition \ref{beta_prop} is understood as the function of $p\in(0,\,1)$, when $n$ and $k$ are fixed.

Proposition \ref{beta_prop} is often met in probabilistic books; e.g., \cite[p. 82]{beta_binom}.

\begin{proposition}\label{main_prop}
Let $n\in \mathbb{N}$, $k\in\{0,\,1,\,\ldots,\,n-1\}$ be fixed and $p\in(0,\,1)$. Then the inequality \eqref{main_ineq} admits the following equivalent representations:
\begin{align}
&\int_{-\infty}^{\infty}\varphi(x)\mathcal{B}_{n-k,\,k+1}\left(\Phi\left(\sqrt{\frac{\varrho}{1-\varrho}}x-\frac{\Phi^{-1}(p)}{\sqrt{1-\varrho}}\right)\right)dx\label{eq:first}\\
&=\int_{0}^{1}\mathcal{B}_{n-k,\,k+1}\left(\Phi\left(\sqrt{\frac{\varrho}{1-\varrho}}\Phi^{-1}(x)-\frac{\Phi^{-1}(p)}{\sqrt{1-\varrho}}\right)\right)dx
\geqslant 1-\gamma,\label{eq:second}
\end{align}
where $\mathcal{B}_{n-k,\,k+1}(\cdot)$ is the cumulative distribution function of the beta random variable.
\end{proposition}

{\sc Note 2}: The same way as Proposition \ref{beta_prop} relates the cumulative distribution functions of binomial and beta distributions, Proposition \ref{main_prop} relates the cumulative distribution function $\mathbb{P}(\mathcal{D}\leqslant k)$, $k\in\{0,\,1,\,\ldots,\,n\}$ from \eqref{main_ineq} to the ones in \eqref{eq:first} or \eqref{eq:second} when $p\in(0,\,1)$. Of course, $\mathcal{B}_{n-k,\,k+1}(\cdot)$ in \eqref{eq:first} and \eqref{eq:second} can be easily replaced by  
$\mathcal{BN}_{n-k,\,k+1,\,0,\,1}(\cdot)$ due to the argument $\Phi(\cdot)$.

We denote $\mathcal{U}(0,\,1)$ the uniform distribution over the interval $(0,\,1)$. Then, the following proposition is correct.

\begin{proposition}\label{prop:dist_eq}
Let $X\sim\mathcal{N}(0,\,1)$, $Y\sim \mathcal{BN}(n-k,\,k+1,\,0,\,1)$, $Z\sim\mathcal{U}(0,\,1)$, $W\sim\mathcal{B}(n-k,\,k+1)$ and suppose that the random variables in pairs $(X,\,Y)$, $(Y,\,Z)$, $(Z,\,W)$ are independent. Then the distribution function in \eqref{eq:first} or \eqref{eq:second}, when $p\in(0,\,1)$, equals to
\begin{align}
&\mathbb{P}\left(\Phi\left(\sqrt{\varrho}X-\sqrt{1-\varrho}Y\right)>p\right)\label{n_and_bn}\\
=&\mathbb{P}\left(\Phi\left(\sqrt{\varrho}\Phi^{-1}(Z)-\sqrt{1-\varrho}Y\right)>p\right)\label{u_and_bn}\\
=&\mathbb{P}\left(\Phi\left(\sqrt{\varrho}\Phi^{-1}(Z)-\sqrt{1-\varrho}\Phi^{-1}(W)\right)>p\right).\label{u_and_b}
\end{align}
\end{proposition}

{\sc Note 3:} Of course, there can be given some other joint distributions' expressions than those provided in \eqref{n_and_bn}, \eqref{u_and_bn}, \eqref{u_and_b}.

\begin{corollary}\label{corollary}
If $k=0$, $n\in\mathbb{N}$ and $X\sim \mathcal{N}(0,1)$, then the left hand-side of the inequality \eqref{main_ineq} is
\begin{align}
&\hspace{-0,5cm}\mathbb{E}\Phi^n\left(\sqrt{\frac{\varrho}{1-\varrho}}X-\frac{\Phi^{-1}(p)}{\sqrt{1-\varrho}}\right)
=\int_{-\infty}^{+\infty}\varphi(x)\Phi^n\left(\sqrt{\frac{\varrho}{1-\varrho}}x-\frac{\Phi^{-1}(p)}{\sqrt{1-\varrho}}\right)\,dx\label{c_first}\\
&\hspace{-0,5cm}=\Phi_{R}\left(-\Phi^{-1}(p),\,\ldots,\,-\Phi^{-1}(p)\right)\label{c_second},
\end{align}
where $\Phi_{R}$ is the Gaussian copula with the correlation matrix
\begin{align*}
R=
\begin{pmatrix}
1&\varrho&\ldots&\varrho\\
\varrho&1&\ldots&\varrho\\
\vdots&\vdots&\ddots&\vdots\\
\varrho&\varrho&\ldots&1
\end{pmatrix}_{n\times n}.
\end{align*}

On top of that, the multivariate density of $\Phi_{R}$ in \eqref{c_second} is
\begin{align}\label{m_density}
\varphi_R:=\frac{\exp\left\{\frac{(1+(n-2)\varrho)\sum_{i=1}^{n}x_i^2-2\varrho\sum_{1\leqslant i<j\leqslant n}x_ix_j}{-2(1-\varrho)(1+(n-1)\varrho)}\right\}}{
\sqrt{(2\pi)^n(1-\varrho)^{n-1}(1+(n-1)\varrho)}},\,
(x_1,\,\ldots,\,x_n)\in\mathbb{R}^n.
\end{align}
\end{corollary}

Corollary \ref{corollary} and its proof (see Section \ref{sec:proofs}) implies
$$
\mathbb{E}\Phi^n\left(-\sqrt{\frac{\varrho}{1-\varrho}}X+\frac{\Phi^{-1}(p)}{\sqrt{1-\varrho}}\right)
=\Phi_{R}\left(\Phi^{-1}(p),\,\ldots,\,\Phi^{-1}(p)\right),
$$
where $X\sim\mathcal{N}(0,\,1)$, and these moments of Vasicek distribution \eqref{Vasicek} are connected to the moments of the probability distribution given by
\begin{align}\label{dist_bin_s1}\nonumber
&\mathbb{P}(X=i)=\\
&\binom{n}{i}\int_{-\infty}^{+\infty}\hspace{-0.2cm}\vphi(x)
\Phi^i\left(\frac{\Phi^{-1}(p)-\sqrt{\varrho}x}{\sqrt{1-\varrho}}\right)
\left(1-\Phi\left(\frac{\Phi^{-1}(p)-\sqrt{\varrho}x}{\sqrt{1-\varrho}}\right)\right)^{n-i}\hspace{0cm}dx,
\end{align}
where $n\in\mathbb{N}$ and $i\in\{0,\,1,\,\ldots,\,n\}$, see \eqref{bin_dist_s}. Indeed, due to the well-known moment-generating function of the binomial distribution, the moment-generating function $M(t)$ of \eqref{dist_bin_s1} is
$$
\mathbb{E}\left(\Phi\left(\sqrt{\frac{\varrho}{1-\varrho}}X-\frac{\Phi^{-1}(p)}{\sqrt{1-\varrho}}\right)+\Phi\left(-\sqrt{\frac{\varrho}{1-\varrho}}X+\frac{\Phi^{-1}(p)}{\sqrt{1-\varrho}}\right)e^t\right)^n,\,t\in\mathbb{R},
$$
where $X\sim\mathcal{N}(0,\,1)$. Notice that $M(\log t)$ is the probability-generating function of the underlying distribution.

\section{Proofs}\label{sec:proofs}
This section provides the proofs for three formulated statements in Section \ref{sec:statements}. Majority of the given proofs are commonly known among researchers or scholars and there is difficult to give any initial source.

\begin{proof}[Proof of Proposition \ref{beta_prop}]
Let us first show that 
$$1-\mathcal{B}_{k+1,n-k}(p)=\mathcal{B}_{n-k,k+1}(1-p).$$ Indeed, 
\begin{align*}
&\mathcal{B}_{n-k,\,k+1}(1-p)=
\frac{\int_{0}^{1-p}u^{n-k-1}(1-u)^{k}du}{B(n-k,\,k+1)}
=-\frac{\int_{1}^{p}(1-x)^{n-k-1}x^{k}dx}{B(n-k,\,k+1)}\\
&=\frac{\int_{p}^{1}x^k(1-x)^{n-k-1}dx}{B(k+1,\,n-k)}
=\frac{B(k+1,\,n-k)-\int_{0}^{p}x^{k}(1-x)^{n-k-1}dx}{B(k+1,\,n-k)}\\
&=1-\mathcal{B}_{k+1,n-k}(p).
\end{align*}

We now aim to prove
$$
1-\mathcal{B}in_{n,p}(k)
=\mathcal{B}_{n-k,k+1}(1-p),
$$
where $\mathcal{B}in_{n,p}(k)$ is considered as a function of $p\in(0,1)$ when $k$ and $n$ are fixed. Let us rewrite 
$$
f(p):=1-\mathcal{B}in_{n,p}(k)
=\sum_{i=k+1}^{n}\binom{n}{i}p^i(1-p)^{n-i},\,
k=0,\,1,\,\ldots,\,n-1.
$$
One may observe that $f(0)=0$, $f(1)=1$ and the derivative
\begin{align*}
&\frac{d\,f(p)}{dp}=\sum_{i=k+1}^{n}\binom{n}{i}
\left(ip^{i-1}(1-p)^{n-i}-(n-i)p^i(1-p)^{n-1-i}\right)\\
&=n\sum_{i=k+1}^{n}\left(\binom{n-1}{i-1}p^{i-1}(1-p)^{n-i}-\binom{n-1}{i}p^i(1-p)^{n-1-i}\mathbbm{1}_{\{i\leqslant n-1\}}\right)\\
&=n\Bigg(\binom{n-1}{k}p^{k}(1-p)^{n-1-k}-\binom{n-1}{k+1}p^{k+1}(1-p)^{n-2-k}\\
&+\binom{n-1}{k+1}p^{k+1}(1-p)^{n-2-k}-\binom{n-1}{k+2}p^{k+2}(1-p)^{n-3-k}+\ldots\\
&+\binom{n-1}{n-1}p^{n-1}(1-p)^{0}\Bigg)=
\frac{n!}{k!(n-k-1)!}p^k(1-p)^{n-1-k}
\end{align*}
is positive for all $p\in(0,\,1)$. Thus, $f(p)$ is the cumulative distribution function over the interval $p\in(0,\,1)$ and its derivative is nothing but the density of the beta distribution with parameters $(k+1,\,n-k)$, i.e.,
$$
\frac{d\,f(p)}{dp}=b_{k+1,\,n-k}(p)=\frac{\Gamma(n+1)}{\Gamma(k+1)\Gamma(n-k)}p^k(1-p)^{n-k-1},\,p\in(0,\,1).
$$
\end{proof}

\begin{proof}[Proof of Proposition \ref{main_prop}]
The integral in \eqref{main_ineq} implies \eqref{eq:first} by Proposition \ref{beta_prop}, while \eqref{eq:first} implies \eqref{eq:second} by the change of variable $\Phi(x)\mapsto x$. 
\end{proof}

\begin{proof}[Proof of Proposition \ref{prop:dist_eq}]
The probability \eqref{n_and_bn} is implied by \eqref{eq:first} observing that
\begin{align*}
&\int_{-\infty}^{+\infty}\varphi(x)\left(\int_{-\infty}^{\sqrt{\frac{\varrho}{1-\varrho}}x-\frac{\Phi^{-1}(p)}{\sqrt{1-\varrho}}}bn_{n-k,\,k+1,\,0,\,1}(y)dy\right)dx\\
&=\int_{-\infty}^{+\infty}
\int_{-\infty}^{\sqrt{\frac{\varrho}{1-\varrho}}x-\frac{\Phi^{-1}(p)}{\sqrt{1-\varrho}}}\varphi(x)bn_{n-k,\,k+1,\,0,\,1}(y)\,dx\,dy\\
&=\mathbb{P}\left(Y<\sqrt{\frac{\varrho}{1-\varrho}}X-\frac{\Phi^{-1}(p)}{\sqrt{1-\varrho}}\right)
=\mathbb{P}\left(\sqrt{\varrho}X-\sqrt{1-\varrho}Y>\Phi^{-1}(p)\right)\\
&=\mathbb{P}\left(\Phi\left(\sqrt{\varrho}X-\sqrt{1-\varrho}Y\right)>p\right),
\end{align*}
when $X$ and $Y$ are independent. The remaining probabilities \eqref{u_and_bn} and \eqref{u_and_b} are implied by the integral in \eqref{eq:second} by the same arguments.
\end{proof}

\begin{proof}[Proof of Corollary \ref{corollary}]
Let $a,\,b\in\mathbb{R}$. Assume the random variables $Y_1,\,\ldots,\,Y_n$ are independent and identically distributed by $\mathcal{N}(0,\,1)$. If $X\sim\mathcal{N}(0,\,1)$ and  $Y_1,\,\ldots,\,Y_n$ are conditionally independent of $X$, then
\begin{align*}
&\mathbb{P}(Y_1<aX+b,\,\ldots,\,Y_n<aX+b)\\
&=\int_{-\infty}^{+\infty}\mathbb{P}(Y_1<aX+b,\,\ldots,\,Y_n<aX+b|X=x)
\varphi(x)\,dx\\
&=\int_{-\infty}^{+\infty}\varphi(x)\Phi^n(ax+b)\,dx=\mathbb{E}\Phi^n(aX+b).
\end{align*}
The equality \eqref{c_first} follows by choosing 
$$
a=\sqrt{\frac{\varrho}{1-\varrho}},\,b=-\frac{\Phi^{-1}(p)}{\sqrt{1-p}},
$$
while the equality \eqref{c_second} is implied observing that
\begin{align*}
&\mathbb{P}\left(\sqrt{1-\varrho}Y_1-\sqrt{\varrho}X<-\Phi^{-1}(p),\,\ldots,\,\sqrt{1-\varrho}Y_n-\sqrt{\varrho}X<-\Phi^{-1}(p)\right)\\
&=\Phi_R\left(-\Phi^{-1}(p),\,\ldots,\,-\Phi^{-1}(p)\right),
\end{align*}
where 
\begin{align*}
%\mathbf{\mu}=\left(0,\,\ldots,\,0\right)_{1\times n},\,
R=
\begin{pmatrix}
1&\varrho&\ldots&\varrho\\
\varrho&1&\ldots&\varrho\\
\vdots&\vdots&\ddots&\vdots\\
\varrho&\varrho&\ldots&1
\end{pmatrix}_{n\times n},
\end{align*}
because
\begin{align*}
&\textit{corr}\left(\sqrt{1-\varrho}Y_i-\sqrt{\varrho}X,\,\sqrt{1-\varrho}Y_j-\sqrt{\varrho}X\right)=
\begin{cases}
&1,\,i=j,\\
&\varrho,\,i\neq j
\end{cases}
\end{align*}
and 
\begin{align*}
\mathbb{E}\left(\sqrt{1-\varrho}Y_i-\sqrt{\varrho}X\right)=0,\, i=1,\,\ldots,\,n.
\end{align*}

The determinant of $R$ is
$$
|R|=(1-\varrho)^{n-1}(1+(n-1)\varrho)
$$
and the inverse matrix of $R$ admits the following representation

\begin{align*}
R^{-1}=\frac{1}{(1-\varrho)(1+(n-1)\varrho)}
\begin{pmatrix}
1+(n-2)\varrho&-\varrho&\ldots&-\varrho\\
-\varrho&1+(n-2)\varrho&\ldots&-\varrho\\
\vdots&\vdots&\ddots&\vdots\\
-\varrho&-\varrho&\ldots&1+(n-2)\varrho
\end{pmatrix}.
\end{align*}
Indeed, it is easy to check that $RR^{-1}=I$, where $I$ is the identity matrix. Then, the multivariate density \eqref{m_density} is implied by the formula
$$
\frac{\exp\left\{-\frac{1}{2}(x_1,\,\ldots,\,x_n)R^{-1}(x_1,\,\ldots,\,x_n)^T\right\}}{\sqrt{(2\pi)^n|R|}},\,
(x_1,\,\ldots,\,x_n)\in\mathbb{R}^n,
$$
see, for example, \cite{Gut}, \cite{GK1}, \cite{GK2}.
\end{proof}

\section{Examples of computation}\label{sec:examples}

In this section, we give two examples that illustrate the discussed estimation of default probability $p$. The required computations are performed with program \cite{Mathematica}.

\begin{example}\label{old_example}
Suppose there are up to $3$ defaults expected with probability $1-\gamma$ out of $800$ obligors which are split into three risk classes: $A,\,B$ and $C$, where $A$ represents the lowest risk and $C$ the highest. Assume the numbers of obligors are $100,\,400,\,300$ and the numbers of expected defaults are up to $0,\,2,\,1$ in risk classes $A$, $B$ and $C$ respectively. We apply Propositions \ref{beta_prop}, \ref{main_prop} and the method of conservatism introduced in \cite{Pluto2006} to estimate the probabilities of default $p_A$, $p_B$ and $p_C$ in risk classes $A$, $B$ and $C$.
\end{example}

The method of conservatism (see \cite{Pluto2006} and the description by end of Section \ref{section_2}) states that $p_A$ should be estimated for the entire portfolio, i.e., $n=800$ and $k=3$ in the considered case. The probability $p_B$ should be estimated for the entire portfolio excluding the class $A$, i.e., $n=700$ and $k=3$ in the considered case. Then, the probability $p_C$ is estimated using $n=300$ and $k=1$ as per the riskiest class $C$.

Using Proposition \ref{beta_prop}, the underlying logic stated in subsection \ref{sub1} and the method of conservatism, we obtain Table \ref{table1}. 
\begin{table}[H]
\centering
\begin{tabular}{|c|c|c|c|c|c|c|}
\hline
$\gamma$&$0.5$&$0.75$&$0.9$&$0.95$&$0.99$&$0.999$\\
\hline
$1-\mathcal{B}^{-1}_{797,\,4}(1-\gamma)$&$0.46\%$&$0.64\%$&$0.83\%$&$0.97\%$&$1.25\%$&$1.62\%$\\
\hline
$1-\mathcal{B}^{-1}_{697,\,4}(1-\gamma)$&$0.52\%$&$0.73\%$&$0.95\%$&$1.10\%$&$1.43\%$&$1.85\%$\\
\hline
$1-\mathcal{B}^{-1}_{299,\,2}(1-\gamma)$&$0.56\%$&$0.90\%$&$1.29\%$&$1.57\%$&$2.19\%$&$3.04\%$\\
\hline
\end{tabular}
\caption{The upper bounds of $p_A$, $p_B$ and $p_C$.}\label{table1}
\end{table}

Note that the numbers in Table \ref{table1} are given in \cite{Pluto2006} too and we replicate them for comparison purposes, especially calculating the quantiles of the underlying distribution given by $F_{n-k,\,k+1,\,\varrho}(y)$.

Suppose the asset correlation $\varrho=12\%$ in Example \ref{old_example}. Then, 
using Proposition \ref{main_prop}, the underlying logic stated in subsection \ref{sub2}, the method of conservatism and the function "FindRoot" in progam \cite{Mathematica} we obtain Table \ref{table2} and Table \ref{table3}.

\begin{table}[H]
\centering
\begin{tabular}{|c|c|c|c|c|c|c|}
\hline
$\gamma$&$0.5$&$0.75$&$0.9$&$0.95$&$0.99$&$0.999$\\
\hline
$F_{797,\,4,\,0.12}^{-1}(1-\gamma)$&$2.61$&$2.34$&$2.09$&$1.94$&$1.67$&$1.36$\\
\hline
$F_{697,\,4,\,0.12}^{-1}(1-\gamma)$&$2.57$&$2.29$&$2.04$&$1.90$&$1.62$&$1.31$\\
\hline
$F_{299,\,2,\,0.12}^{-1}(1-\gamma)$&$2.55$&$2.25$&$1.98$&$1.82$&$1.52$&$1.19$\\
\hline
\end{tabular}
\caption{The quantiles of distribution which cumulative distribution function is $F_{n-k,\,k+1,\,\varrho}(y)$.}\label{table2}
\end{table}

\begin{table}[H]
\centering
\begin{tabular}{|c|c|c|c|c|c|c|}
\hline
$\gamma$&$0.5$&$0.75$&$0.9$&$0.95$&$0.99$&$0.999$\\
\hline
$\Phi(a)$&$0.71\%$&$1.41\%$&$2.49\%$&$3.41\%$&$5.88\%$&$10.08\%$\\
\hline
$\Phi(b)$&$0.80\%$&$1.58\%$&$2.76\%$&$3.77\%$&$6.43\%$&$10.91\%$\\
\hline
$\Phi(c)$&$0.84\%$&$1.75\%$&$3.18\%$&$4.41\%$&$7.67\%$&$13.13\%$\\
\hline
\end{tabular}
\caption{The upper bounds of $p_A$, $p_B$ and $p_C$ under the influence of systematic factor. Here $a=-\sqrt{1-\varrho}F_{797,\,4,\,0.12}^{-1}(1-\gamma)$, $b=-\sqrt{1-\varrho}F_{697,\,4,\,0.12}^{-1}(1-\gamma)$, $c=-\sqrt{1-\varrho}F_{299,\,2,\,0.12}^{-1}(1-\gamma)$ as provided in Table \ref{table2}.}\label{table3}
\end{table}

The provided numbers in Table \ref{table3} match the corresponding ones in \cite{Pluto2006} except few cases caused by rounding errors in the fourth decimal place.

\begin{example}\label{new_example}
Suppose there are up to $7$ defaults expected with probability $1-\gamma$ out of $1500$ obligors which are split in four risk classes: $A,\,B$, $C$ and $D$ where $A$ represents the lowest risk and $D$ the highest. Assume the numbers of obligors are $400,\,700,\,250,\,150$ and the numbers of expected defaults are up to $2,\,1,\,3,\,1$ in risk classes $A$, $B$, $C$ and $D$ respectively. We apply Propositions \ref{beta_prop}, \ref{main_prop} and the method of conservatism introduced in \cite{Pluto2006} to estimate the probabilities of default $p_A$, $p_B$ $p_C$ and $p_D$ in risk classes $A$, $B$, $C$ and $D$.
\end{example}

Using Proposition \ref{beta_prop}, the underlying logic stated in subsection \ref{sub1} and the method of conservatism, we obtain Table \ref{table4}.

\begin{table}[H]
\centering
\begin{tabular}{|c|c|c|c|c|c|c|}
\hline
$\gamma$&$0.5$&$0.75$&$0.9$&$0.95$&$0.99$&$0.999$\\
\hline
$\mathcal{B}^{-1}_{1493,\,8}(1-\gamma)$&$0.51\%$&$0.65\%$&$0.78\%$&$0.87\%$&$1.06\%$&$1.30\%$\\
\hline
$\mathcal{B}^{-1}_{1095,\,6}(1-\gamma)$&$0.52\%$&$0.67\%$&$0.84\%$&$0.95\%$&$1.19\%$&$1.49\%$\\
\hline
$\mathcal{B}^{-1}_{396,\,5}(1-\gamma)$&$1.17\%$&$1.56\%$&$1.99\%$&$2.27\%$&$2.87\%$&$3.65\%$\\
\hline
$\mathcal{B}^{-1}_{149,\,2}(1-\gamma)$&$ {\bm{ 1.12}}\%$&$1.78\%$&$2.57\%$&$3.12\%$&$4.34\%$&$5.99\%$\\
\hline
\end{tabular}
\caption{The upper bounds of $p_A$, $p_B$, $p_C$ and $p_D$.}\label{table4}
\end{table}

Notice that $\mathcal{B}^{-1}_{396,\,5}(1/2)>\mathcal{B}^{-1}_{149,\,2}(1/2)$ in Table \ref{table4} and (see \cite[Footnote 6]{Pluto2006}) ''... this is not a desirable effect, a possible – conservative – work-around could be to increment the number of defaults in
grade $D$ up to the point where $p_D$ would take on a greater value than $p_C$ ...''.  

Suppose the asset correlation $\varrho=12\%$ in Example \ref{new_example}. Then, 
using Proposition \ref{main_prop}, the underlying logic stated in subsection \ref{sub2}, the method of conservatism and the function "FindRoot" in progam \cite{Mathematica} we obtain Table \ref{table5} and Table \ref{table6}.

\begin{table}[H]
\centering
\begin{tabular}{|c|c|c|c|c|c|c|}
\hline
$\gamma$&$0.5$&$0.75$&$0.9$&$0.95$&$0.99$&$0.999$\\
\hline
$F_{1493,\,8,\,0.12}^{-1}(1-\gamma)$&2.57&$2.31$&$2.07$&$1.93$&$1.67$&$1.37$\\
\hline
$F_{1095,\,6,\,0.12}^{-1}(1-\gamma)$&$2.57$&$2.30$&$2.06$&$1.92$&$1.65$&$1.35$\\
\hline
$F_{396,\,5,\,0.12}^{-1}(1-\gamma)$&$2.27$&$2.00$&$1.75$&$1.61$&$1.33$&$1.02$\\
\hline
$F_{149,\,2,\,0.12}^{-1}(1-\gamma)$&$ 2.30$&$1.98$&$1.71$&$1.54$&$1.24$&$0.91$\\
\hline
\end{tabular}
\caption{The quantiles of distribution which cumulative distribution function is $F_{n-k,\,k+1,\,\varrho}(y)$.}\label{table5}
\end{table}

\begin{table}[H]
\centering
\begin{tabular}{|c|c|c|c|c|c|c|}
\hline
$\gamma$&$0.5$&$0.75$&$0.9$&$0.95$&$0.99$&$0.999$\\
\hline
$\Phi(a)$&$0.79\%$&$1.51\%$&$2.59\%$&$3.49\%$&$5.58\%$&$9.90\%$\\
\hline
$\Phi(b)$&$0.79\%$&$1.53\%$&$2.64\%$&$3.58\%$&$6.06\%$&$10.23\%$\\
\hline
$\Phi(c)$&$1.64\%$&$3.04\%$&$5.01\%$&$6.60\%$&$10.61\%$&$16.87\%$\\
\hline
$\Phi(d)$&$ {\bm{ 1.56}}\%$&$3.13\%$&$5.45\%$&$7.36\%$&$12.21\%$&$19.76\%$\\
\hline
\end{tabular}
\caption{The upper bounds of $p_A$, $p_B$, $p_C$ and $p_D$ under the influence of systematic factor. Here $a=-\sqrt{1-\varrho}F_{1493,\,8,\,0.12}^{-1}(1-\gamma)$, $b=-\sqrt{1-\varrho}F_{1095,\,6,\,0.12}^{-1}(1-\gamma)$, $c=-\sqrt{1-\varrho}F_{396,\,5,\,0.12}^{-1}(1-\gamma)$, $d=-\sqrt{1-\varrho}F_{149,\,2,\,0.12}^{-1}(1-\gamma)$ as provided in Table \ref{table5}.}\label{table6}
\end{table}

Notice that $F_{396,\,5,\,0.12}^{-1}(1/2)>F_{149,\,2,\,0.12}^{-1}(1/2)$ in Table \ref{table5} and consequently the corresponding upper bounds of $p_C$ and $p_D$ in Table \ref{table6} maintain the upper bound reversal.

\section{Concluding remarks}
As stated, this survey article gives a detailed probabilistic overview of two methods for the upper bound of default probability. The provided insights reveal the important role played by the beta-normal distribution. However, the beta-normal distribution appears to be little studied, compared to the voluminous literature for the separate normal or beta distributions. It would be of interest to get any closed-form of the inverse of $F_{\al,\,\be,\,\varrho}(p)$ (see \eqref{cum_dist}) in terms of a superposition of $\Phi^{-1}_{\mu,\,\sigma^2}(\cdot)$ and $\mathcal{B}^{-1}_{\al,\,\be}(\cdot)$, which possibly would include studying the cumulative distribution function $\Phi_{\mu,\,\sigma^2}\left(a\Phi^{-1}_{\tilde{\mu},\,\tilde{\sigma}^2}(x)+b\right)$ when $a,\,b\in\mathbb{R}$ and $x\in(0,\,1)$.

\section{Acknowledgments}

The author is thankful to Arvydas Karbonskis for his feedback on the draft version of this article and also to Dirk Tasche for pointing to the reference \cite{Dirk} and giving several other valuable comments.

\bibliographystyle{elsarticle-harv} 
\bibliography{cas-refs}

\begin{flushleft}
Andrius Grigutis\\
Institute of Mathematics\\
Faculty of Mathematics and Informatics, Vilnius University\\
Naugarduko 24, LT-03225 Vilnius, Lithuania\\
andrius.grigutis@mif.vu.lt
\end{flushleft}

\end{document}